\documentclass[12pt]{article}
\usepackage[utf8]{inputenc}
\usepackage{natbib}
\usepackage{eurosym}
\usepackage{amstext,amsthm}
\usepackage{amsmath}
\usepackage{amssymb,mathtools}
\usepackage[nointegrals]{wasysym} 
\usepackage{mathrsfs}
\usepackage{hyperref}
\usepackage{graphicx}
\usepackage[figuresleft]{rotating}

\usepackage{tikz}

\newtheorem{proposition}{Proposition}[section]

\newtheorem{lemma}[proposition]{Lemma}
\newtheorem{theorem}{Theorem}
\theoremstyle{definition}
\newtheorem{remark}[proposition]{Remark}
\newtheorem{assumption}{Assumption}
\newtheorem{definition}[proposition]{Definition}

\newcommand{\diag}{\text{diag}}

\usepackage[normalem]{ulem}
\usepackage{color}

\usepackage{times}

\makeatletter
\renewcommand{\@fnsymbol}[1]{\ensuremath{%
    \ifcase#1\or 1\or 2\or 3\or
    \mathsection\or \mathparagraph\or \|\or 1\or
    2\or 3 \else\@ctrerr\fi}}
\makeatother
\begin{document}

\setcounter{page}{1} \title{\LARGE \color{black}A central limit theorem
  concerning uncertainty in estimates of individual admixture}
\author{\sc Peter Pfaffelhuber\thanks{Abteilung f\"ur Mathematische
    Stochastik, Albert-Ludwigs University of Freiburg,
    Ernst-Zermelo. 1, D - 79104 Freiburg, Germany, e-mail:
    p.p@stochastik.uni-freiburg.de}, \sc Angelika
  Rohde\thanks{Abteilung f\"ur Mathematische Stochastik,
    Albert-Ludwigs University of Freiburg, Ernst-Zermelo. 1, D - 79104
    Freiburg, Germany, e-mail:
    angelika.rohde@stochastik.uni-freiburg.de}}

\date{\today}

\maketitle

\begin{abstract}
  \noindent
  The concept of individual admixture (IA) assumes that the genome of individuals is composed of alleles inherited from $K$ ancestral populations. Each copy of each allele has the same chance $q_k$ to  originate from population $k$, and together with the allele frequencies $p$ in all populations at all $M$ markers, comprises the admixture model. Here, we assume a supervised scheme, i.e.\ allele frequencies $p$ are given through a reference database of size $N$, and $q$ is estimated via maximum likelihood for a single sample. We {\color{black} study laws of large numbers and central limit theorems describing }effects of finiteness of both, $M$ and $N$, on the estimate of $q$. We recall results for the effect of finite $M$, and provide a central limit theorem for the effect of finite $N$, {\color{black} introduce a new way to express the uncertainty in estimates in standard barplots}, give simulation results, and discuss applications in forensic genetics.
\end{abstract}


\section{Introduction}
When dealing with genetic data of a human sample, the  ancestry of the sample donor is of great interest in fields such as population history (see e.g.\ \citealp{Rosenberg2002}), genome wide association studies, admixture mapping (e.g.\ \citealp{pmid15088268}), and forensic genetics (see e.g.\ \cite{Kidd2021}
for a recent contribution). Using a reference database of allelic frequencies in various ancestral populations, we distinguish between all-or-nothing classifiers, and admixture models. In the former, the sample is classified into one (and only one) of the reference populations, in the latter, the task is to estimate the individual admixture (IA) proportions of all populations. 

We are going to analyse admixture models. Today, they are implemented in software such as {\sc STRUCTURE} \citep{Pritchard2000}, or its successor {\sc ADMIXTURE} \citep{Alexander2009}, or various others (reviewed in \citealp{pmid25937887}). The underlying probabilistic model comes with $K$ ancestral populations, where each population comes with its own allele frequencies at $M$ ancestry informative markers (AIMs, most frequently SNPs), \sloppy collected in some array $p = (p_{kim})_{k=1,...,K, i=1,...,I, m=1,...,M}$, where $p_{kim}$ is the frequency of allele $i$ in population $k$ at locus $m$. In this model, it is assumed that every individual with diploid genotype $x = (x_{im})_{i=1,...,I, m=1,...,M}$ (storing the number -- 0, 1 or 2 -- of copies for each of the $I$ alleles at all $M$ loci) is a mixture of the ancestral populations, i.e.\ each allele originates in one of the $K$ ancestral populations, modeled by the individual admixture $q$. The software -- using a Bayes approach -- either simultaneously estimates $p$ and $q$ (unsupervised setting) or uses ancestry information of a large part of the data and estimates $q$ for the remaining samples (supervised setting). In the sequel, we will only study the (analytically simpler) supervised setting. More precisely, we will assume that allele frequencies $p$ are given from an external source of non-admixed samples (i.e.\ the reference database), and we maximize the likelihood of individual admixture of another (putatively admixed) individual (test sample) with respect to~$q$.

Ideally, the $M$ AIMs used to infer IA have large frequency differentials between populations \citep{Rosenberg2005}. However, for real populations, such markers are not always available, leading to a variance in the estimation procedure of IA (see e.g.\ \citealp{PfaffLong2004}). However, other sources of uncertainty must not be neglected: (i) Note that $p$ is only estimated from a finite reference database of samples, e.g.\ from public sequencing projects. The resulting variance of $p$ feeds back into the estimation of $q$. (ii) (Some ancestry proportion of) the test sample may come from a population not represented in the reference database. Still, a ML estimator is found, but it is biased by the restricted reference database. 

Let us give a brief introduction to the literature dealing with the admixture model: While early papers such as \cite{Elston1971}, which is based on a likelihood framework by  \cite{Krieger1965}, have phenotypic data from a hybrid population in mind, the interest in the same model increased at the time when genetic data was available; see \cite{Chakraborty1986} for an overview of the model. While various studies assume population data from $K=2$ ancestral and one hybrid (admixed) population, \cite{Hanis1986} treats the same case we are interested in, where the individual admixture of a single admixed individual is to be estimated. Apparently, the reliability of the estimates is of fundamental importance. As \cite{DiversAllison2011} put it, measurement errors, leading to uncertainty in the estimate $\hat q$, occur since (1) the set of genetic markers is finite and are mostly not perfectly ancestral informative, (2) ancestral population allele frequencies are not known precisely, and (3) the number of ancestral populations that contributed genetically to the test sample is not always well known. As for (1), several approaches have led to algorithms for selecting good markers, mostly based on differentials in allele frequencies or derived quantities (see e.g.\ \citealp{Rosenberg2003, Rosenberg2005}). The variance due to (1) was studied previously to some extent; see e.g.\ \cite{Millar1991, PfaffLong2004}. Other studies \citep{pmid17507670, pmid19339787, DiversAllison2011}  are interested in the effect of the uncertaintly in $\hat q$ on downstream procedures (e.g.\ association tests). Our contribution here is a more analytic approach to the admixture model: We re-derive and re-formulate variance estimates from \cite{PfaffLong2004} due to~(1) in Remark~\ref{rem:Minfty}, give the first theoretical treatment of the variance induced by (2) -- see Theorem~\ref{T1}, and use simulations in order to study the bias due to (3). Formally, the derivation of Theorem~\ref{T1} (and the calculations in Remark~\ref{rem:Minfty}) are an application of the theory of $M$-estimation; see e.g.\ \cite{vanderGeer2009}. However, we keep all our calculations self-contained and only use standard results from calculus. 
On our way, we state precise Assumptions~\ref{ass0} and~\ref{ass1}, in order to guarantee existence and uniqueness of the maximum-likelihood (ML) estimator. Another interesting result is a novel iterative scheme for efficient computation of the ML-estimator (Theorem~\ref{T2}). Although this scheme is not the fastest approach for estimation of IA (as compared e.g.\ to the block-relaxation method of \citealp{Alexander2009}), it is straight-forward to implement in just a few lines of code. We are going to use this scheme for simulations and examples from a forensic database in Section~\ref{S3}. {\color{black} For the reader interested in applications of our results, we establish in Section~\ref{ss:31} a way to display the estimated variance in the estimation of IA in standard barplots. Apart from simulations,} we consider two marker sets used in forensic genetics, with a reference database available on the homepage of the software {\sc Snipper} \citep{Snipper2007}; see also Section~\ref{S33} for more details. We start by introducing some notation which is required in the sequel.

\begin{remark}[Notation]
  We denote by $x \in \mathbb R^n$ (for some $n=2,3,...$) a column   vector, and by $x^\top$ the corresponding row vector, i.e.\ $\top$   denotes a transposition. The standard scalar product of   $p,q\in\mathbb R^n$ is   $\langle q,p\rangle = \langle p,q\rangle := p^\top \cdot q$. We use   $E_n$ for the identity matrix with $n$ rows and columns and $\diag(x)$ the diagonal matrix with $x_1,...,x_n$ on the diagonal.
  \\
  For $K=2,3,...$, we frequently need $\mathbb S := \mathbb S_K$, the $K-1$-dimensional   simplex, i.e.\ $$\mathbb S := \{q = (q_1,...,q_K) \in \mathbb R_{\geq 0}: q_1 + \cdots + q_K = 1\}.$$ The inner set of $\mathbb S$ is  $\mathbb S^\circ := \{q = (q_1,...,q_k) \in \mathbb S: q_1, ..., q_K >0\}$. For $I=2,3,...$ and $p\in \mathbb S_I$, we recall the multinomial distribution with $n$ trials and success probabilities $p = (p_1,...,p_I)$. Here, $X = (X_1,...,X_I) \sim \text{Multi}(n,p)$ for some $n=1,2,...$ and $p\in\mathbb S_I$, iff   for $x=(x_1,...,x_I)$ with $x_1,...,x_I\geq 0$ and $x_1 + \cdots + x_I = n$ $$ \mathbb P(X = x) = \binom nx p^x, \qquad \binom nx := \frac{n!}{x_1! \cdots x_I!}, \qquad p^x := p_1^{x_1} \cdots p_I^{x_I}.$$ Further, for some smooth $f: \mathbb R^n \to \mathbb R$, let  $\nabla f(x) \in \mathbb R^n$ be the gradient of $f$, evaluated at $x$, with entries $\frac{\partial f(x)}{\partial x_i}, i=1,...,n$, and by $\nabla^2 f(x) \in \mathbb R^{n\times n}$ the (symmetric) Hessian of $f$, with entries  $\frac{\partial^2 f(x)}{\partial x_i \partial x_j}, i,j=1,...,n$. On a locally compact metric space $E$, we denote by $\mathcal M_n(E)$ the set of measures (defined on the  Borel-sigma-field of $E$) with total mass~$n$, and on the set $\mathcal M_1(E)$ of probability measures, we denote weak  convergence by $\Rightarrow$.
\end{remark}

\section{Model and main results}
\sloppy In \cite{Pritchard2000}, \cite{Tang2005}, and elsewhere, the following model was used: Assuming a diploid species, there are $K\geq 2$ populations and $M$ markers in {\color{black}(multi-locus)} linkage equilibrium, and at marker $m$, there are $I_m$ possible alleles, leading to a total of $\binom{I_m+1}{2}$ possible genotypes (e.g.\ for a SNP, we have $I_m=4$ and genotypes $\{AA, AC, AG, AT, CC, CG, CT, GG, GT, TT\}$, where we do not distinguish phase, i.e.\ $AC$ and $CA$ are indistinguishable; think of $1 \equiv A, 2\equiv C, 3\equiv G, 4\equiv T$) and for locus $m$, each {\color{black} allele} $i=1,...,I_m$ in population $k$ has frequency $p_{kim}$ and $\sum_i p_{kim} = 1$, for $k=1,...,K, m=1,...,M$. In order to simplify notation in the sequel, we assume that $I = I_m, m=1,...,M$, i.e.\ each marker has the same number of alleles. Now, the genotype of a single individual is characterized by $x = (x_1,...,x_M)$ with $x_m = (x_{im})_{i=1,...,I} \in \{0,1,2\}^{I}$ with $\sum_i x_{im} = 2$ for $m=1,...,M$. In addition, the genome of the individual is a mixture of the $K$ (non-admixed) classes according to some probability vector $q \in \mathbb S$. This means that each copy of each locus has a{\color{black}n independent} chance $q_k$ to originate from population $k$ for $k=1,...,K$. We will call $q$ the {\em individual admixture (IA)} of $x$. This leads to the following definition of the statistical model for data from a single individual we treat here:

\begin{definition}[Admixture model] \label{def:admixtureModel} \mbox{} Let $p = (p_{\cdot \cdot m})_{m=1,...,M} = (p_{k\cdot m})_{k=1,...,K,    m=1,...,M} \in \mathbb S_{I}^{K\times M}$ (i.e.\ for all $k\in K, m\in M$, we have $p_{k\cdot m} \in \mathbb S_{I}$) and $q \in \mathbb S$ be given. Then, the probability of obtaining the (test) data $X = (X_{\cdot 1}, ..., X_{\cdot M}) \in \{0,1,2\}^{I\times M}$ is given by $X_{\cdot m} \sim \text{Multi}(2, (\langle q, p_{\cdot    im}\rangle)_{i=1,...,I})$ (the multinomial distribution with $2$  trials and probabilities $(\langle q, p_{\cdot im}\rangle)_{i=1,...,I}$), independently for  all $m=1,...,M$. Thus, the log-likelihood (scaled by $2M$) is given by
  \begin{equation}
    \label{eq:L|P}
    \begin{aligned}
      \ell(p, q|x) & := \frac 1{2M}\log \mathbb P_{p,q} (X_1 =
      x_1,...,X_M = x_M) \\ & := \frac 1{2M} \sum_{m=1}^M      \log(\text{Multi}(2, (\langle q, p_{\cdot im}\rangle)_{i=1,...,I})(x_{\cdot m})) \\ & = \frac 1{2M}      \sum_{m=1}^M \Big(\log\binom 2{x_{\cdot m}} +      \sum_{i=1}^I x_{im} \log(\langle q, p_{\cdot im}\rangle)\Big)\\      & = C_x + \int z \log(\langle q, \alpha\rangle)      \mu_{p,x}(d\alpha, dz),
    \end{aligned}
  \end{equation}
  where $C_x$ is a constant only depending on $x$ (but not on $p,q$), and
  \begin{align}\label{eq:emp}
    \mu_{p,x} & = \frac 1{2M}               \sum_{m=1}^M\sum_{i=1}^I \delta_{(p_{\cdot im}, x_{im})} \in \mathcal M_{I/2}([0,1]^{K}                \times \{0,1,2\})
  \end{align}
  is an empirical measure.
\end{definition}

\begin{remark}[Similar models and the bi-allelic case]
  \begin{enumerate}
     \item In \cite{Pritchard2000}, the above model coincides with the {\em model with admixture}. Restricting to the possibilities of $q = e_k, k=1,...,K$ (the unit vektors in $\mathbb R^K$), a special case is the {\em model without admixture}, which we also denote by an all-or-nothing classifier. In both models, markers are assumed to be in {\color{black}(multi-locus) linkage equilibrium within ancestral populations}. Note that this restriction has been relaxed by \cite{Falush2003}, who build {\em admixture linkage disequilibrium} into the model, where markers forms a Markov chain rather than an independent family within each ancestral population. A rather novel extension is the {\em recent-admixture model}, where the two copies of each marker do not segregate independently, but originate from the two parents which come with their own IA; see \cite{pmid34735936}.\\
     \item As for generalizations, we note that the likelihood if the number of alleles $I_m$ depends on the marker $m$ is just a slight modification of \eqref{eq:L|P}. 
     In addition, note that it is also straight-forward to extend any of these models to any level of ploidy (as in the current version of {\sc STRUCTURE}), by replacing 2 by the ploidy level in \eqref{eq:L|P}. 
  \item The special case of bi-allelic SNP-markers arises frequently in applications. For this reason, we will spell out this special case in Remarks~\ref{rem:bi1}, \ref{rem:bi2}, \ref{rem:bi3}, \ref{rem:bi4}. In this case, some simplifications of our results apply. 
  \end{enumerate}
\end{remark}

~

\noindent
The first assumption states that if we observe allele $i$ at locus $m$ in our test  data, there is at least one population $k$ where the same allele has also been observed. {\color{black} Otherwise, the only reasonable conclusion is that the reference populations do not suffice to correctly estimate the individual admixture of the test data. In other word, the assumption is a necessary condition for a fit of model and data.}

\begin{assumption}\label{ass0}
  If $x_{im} > 0$, there exists some
  $k$ with $p_{kim} > 0$ or, equivalently,
  \begin{align*}
    \mu_{p,x}\{(\alpha,z): z>0, \alpha = 0\} = 0.  
  \end{align*}
\end{assumption}

\begin{remark}
  {\color{black} Note that Assumption~\ref{ass0} prevents the likelihood-function from vanishing. Precisely, if} 
  Assumption~\ref{ass0} holds, and if $q\in\mathbb S^\circ$, we have that $z \log(\langle q, \alpha\rangle) > -\infty$ for  $\mu_{p,x}$-almost all $(\alpha,z)$. In particular,  $\ell(p,q|x) > -\infty$ for all $q\in\mathbb S^\circ$, and extrema  of $q\mapsto \ell(p,q|x)$ can be found by differentiation.  
\end{remark}

\subsection{Estimating $q$ via Maximum Likelihood}
In this section, our goal is to compute the
Maximum-Likelihood-estimator (or ML-estimator) for $q$, i.e.\ we need
to find the maximizer $\hat q$ of $q\mapsto \ell(p,q|x)$ (for given
$p,x$). We start by computing the first two derivatives, i.e.\ the
gradient and Hessian of $q\mapsto \ell(p,q|x)$. We obtain
from \eqref{eq:L|P}
\begin{align}
  \label{eq:nablaell}
  \nabla \ell(p,q|x) & = \int z\frac{\alpha}{\langle q,\alpha\rangle}
                       \mu_{p,x}(d\alpha,dz),
  \\ \label{eq:nabla2ell}
  \nabla^2 \ell(p,q|x) & = -\int z \frac{\alpha \cdot \alpha^\top}{\langle q,\alpha\rangle^2}
                         \mu_{p,x}(d\alpha,dz) =: -\Sigma_x^{-1}.   
\end{align}

\noindent
Computing the ML-estimator of $q$ is actually an exercise of applying
the theory of Lagrange multipliers (see e.g.\ supplement of
\citealp{Trench2013} for an online resource), used for maximization
under equality constraints; see Lemma~\ref{l:1} below. In addition,
note that the right-hand-side of \eqref{eq:nabla2ell} implies that
$q\mapsto \ell(p,q|x)$ is concave, which helps to give sufficient
conditions for a maximum. For uniqueness of the maximum,
{\color{black}we assume that no reference population is admixed from
  other reference populations; see also Section~2.4
  of~\cite{pmid33034736}:}

\begin{assumption}\label{ass1}
  Let $s\in\mathbb R^K$ with $\langle s,1\rangle = 0$. Then, there is
  $i$ and $m$ such that $\langle s, p_{\cdot im}\rangle \neq 0$. In
  other words,
  \begin{align*}
    \text{$\mu_{p,x}\{(\alpha,z): \langle s, \alpha\rangle \neq 0\}>0$ for all $s\in\mathbb R^K$
    with $\langle s,1\rangle = 0$,}
  \end{align*}
  
\end{assumption}

\begin{remark}
  In order to explain Assumption~\ref{ass1} better, let us make an example where it does not hold: If $K=3$ and $a\in [0,1]$ such that $p_{3im} = a p_{1im} + (1-a)p_{2im}$ for all $i,m$, we see that population $3$ is a mixture of population~1 (contribution $a$) and population~2 (contribution $1-a$). It is no surprise that the ML-estimator in this example is not unique, since we e.g.\ cannot  distinguish between a mixture of populations~1 and~2 with $q = (a,1-a,0)$ and a sample from population~3, i.e.\ $q = (0,0,1)$. In the context of the above assumption, consider $s =  (a,1-a,-1)$. Then, $\langle s,1\rangle = 0$, and $\langle s, p_{\cdot im}\rangle = 0$ for all $i,m$, i.e.\ the assumption does not hold. \\
  More generally, assume that there exists $s\in\mathbb R^K$ with $\langle s,1\rangle = 0$ and $\langle s,\alpha\rangle = 0$ for $\mu_{p,x}$-almost-all $\alpha$. In this case, for $q\in\mathbb S^\circ$ and $h$ small enough for $q + hs \in \mathbb S^\circ$, we find $\langle q,\alpha\rangle = \langle q + hs,\alpha\rangle$ for
  $\mu_{p,x}$-almost-all $\alpha$, implying
  $\ell(p,q|x) = \ell(p,q + hs|x)$ and the likelihood-curve is flat around $\ell(p,q|x)$ in direction of $s$. So, if Assumption~\ref{ass1} does not hold, we cannot hope for uniqueness of the ML-estimator.
\end{remark}

\noindent
{\color{black}The next lemma assures existence and uniqueness of the Maximum Likelihood estimator. If this estimator $\hat q$ has $\hat q_k>0$ for all $k=1,...,K$, it can be found be differentiation.}

\begin{lemma}[ML-estimator]\label{l:1}
  Let Assumption~\ref{ass0} hold.  Let  $x = (x_{im})_{i=1,...,I, m=1,...,M} \in \{0,1,2\}^{I \times M}$ and $p\in \mathbb S_{I}^{K\times M}$. The function $\ell: \mathbb S \to \mathbb R, q\mapsto \ell(p,q|x)$ is concave. If Assumption~\ref{ass1} holds, it is even strictly concave. In this case, $\ell(p,.|x)$ has a unique maximum, and equals $\hat q$, the ML-estimator for $q$. Moreover, under Assumption~\ref{ass1}, $q^\ast\in\mathbb S^\circ$ maximizes $\ell(p,.|x)$ if and only if $\nabla\ell(p,q^\ast|x) = 1$.
\end{lemma}

\begin{proof}
  From \eqref{eq:nabla2ell}, we see that $\nabla^2\ell(p,q|x)$ is
  non-positive definite for all $q$, which already implies concavity
  of $\ell$. For strict concavity, let $q\in\mathbb S^\circ$ and
  $s \in \mathbb R^K$ be such that $q + hs \in \mathbb S^\circ$ for
  small $h$, which implies that $\langle s,1\rangle = 0$. Then, if
  Assumption~\ref{ass1} holds,
  \begin{align*}
    \ell(p,q + hs|x) & = \ell(p,q|x) + h\nabla \ell(p,q|x) \cdot s + \tfrac 12 h^2
                       \nabla^2 \ell(p,q|x) + o(h^2)
    \\ & = \ell(p,q|x) + h\nabla \ell(p,q) \cdot s - \tfrac 12 h^2
         \int z\frac{\langle s,\alpha\rangle^2}{\langle q,\alpha\rangle^2} \mu_{p,x}(d\alpha, dz) + o(h^2)
    \\ & < \ell(p,q|x) + h\nabla \ell(p,q) \cdot s + o(h^2),
  \end{align*}
  implying strict concavity of $\ell(p,.|x)$. For such a function, defined on a compact set, it is well-known that a unique (local and  global) maximum exists; see e.g.\ \cite{Trench2013}, Section~5.
  \\
  Now, let $q^\ast\in\mathbb S^\circ$ maximize $q\mapsto \ell(p,q|x)$.
  We use the theory of Lagrange multipliers and impose the restriction
  $\langle q,1\rangle = 1$. Then, $q^\ast$ solves, for some
  $\lambda \in \mathbb R$, using \eqref{eq:nablaell},
  \begin{align*}
    & \nabla \ell(p,q^\ast|x) = \int z\frac{\alpha}{\langle q^\ast, \alpha\rangle}
      \mu_{p,x}(d\alpha,dz)= \lambda,
      \qquad \qquad  \langle q^\ast, 1\rangle = 1.
  \end{align*}
  From this, we can eliminate $\lambda$, since
  \begin{align*}
    \lambda = \langle q^\ast, \lambda\rangle = \int z\frac{\langle q^\ast, \alpha\rangle}{\langle q^\ast, \alpha\rangle}
    \mu_{p,x}(d\alpha, dz) = \frac{1}{2M} \sum_{m=1}^M \sum_{i=1}^I x_{im} = 1
  \end{align*}
  i.e.\ $q^\ast$ solves $\nabla \ell(p,q^\ast|x) = 1$. Next, if
  $q^\ast\in \mathbb S^\circ$ satisfies $\nabla\ell(p,q^\ast|x) = 1$
  and if Assumption~\ref{ass1} holds, then, for $s\in\mathbb R^K$ with
  $\langle s,1\rangle =0$ and $h$ small,
  \begin{align*}
    \ell(p,q^\ast + hs|x) & = \ell(p,q^\ast|x) + h \langle 1, s\rangle + \tfrac 12 h^2 \int
                       z \frac{\langle s,\alpha\rangle^2}{\langle q^\ast, \alpha\rangle^2}
                       \mu_{p,x}(d\alpha, dz) + o(h^2)
                       < \ell(p,q^\ast|x),
  \end{align*}
  i.e.\ $q^\ast$ is a local maximum of $\ell(p,.|x)$. This finishes the
  proof.
\end{proof}

\begin{remark}[The bi-allelic case~1]\label{rem:bi1}
  Frequently, SNP-markers are bi-allelic, and we briefly translate our setting and Lemma~\ref{l:1} to this scenario. We will use $p_{km}$ and $1-p_{km}$ for the allele frequencies of both alleles (called {\color{black}alleles $A$ and $B$} in the sequel) in population $k$ at marker $m$, and $x_m$ for the number of occurrences of the $A$ allele. (This is a slight abuse of notation.) We note that our results do not depend on switching $A$ and $B$ allele. In the bi-allelic case, the empirical measure $\mu_{p,x}$ from \eqref{eq:emp} reads 
  \begin{align}
    \mu_{p,x} & = \frac{1}{M} \sum_{m=1}^M \delta_{(p_{\cdot m}, x_m)} \in \mathcal M_1([0,1]^K \times \{0,1,2\}).\label{eq:bimeas1}
  \end{align}
  (Note that the $B$ allele does not appear on the right hand side.) Using this measure, we write for the log-likelihood from \eqref{eq:L|P} 
  \begin{equation}
  \label{eq:ellbi}
    \begin{aligned}
    2\ell(p,q|x) & = C_x + \int (z \log(\langle q, \alpha\rangle) + (2-z) \log(\langle q, 1-\alpha\rangle)) \mu_{p,x}(d\alpha, dz).
  \end{aligned}
  \end{equation}
  From Lemma~\ref{l:1}, we then see -- provided Assumption~\ref{ass0}, which is here
  \begin{align*}
  \mu_{p,x}\{(\alpha,z): z>0, \alpha = 0 \text{ or } z<2, \alpha=1\} = 0,
  \end{align*}
  and~Assumption~\ref{ass1} hold -- that $q^\ast \in \mathbb S^\circ$ is the unique global maximum of $\ell(p,.|x)$ if and only if 
  \begin{align*}
      \nabla^2 \ell(p,q|x) = \tfrac 12 \int \Big(z \frac{\alpha}{\langle q, \alpha\rangle} + (2-z)\frac{1-\alpha}{\langle q, 1-\alpha\rangle}\Big) \mu_{p,x}(d\alpha, dz) = 1.
  \end{align*}
\end{remark}

\subsection{A central limit theorem for the ML-estimator}
\label{S23}
As \cite{DiversAllison2011} put it, the measurement error in IA estimates $q$ can have several reasons: 
\begin{enumerate}
    \item[(1)]  The set of AIMs has imperfectly known ancestral population allele frequencies. 
    \item[(2)] The historical knowledge about ancestral populations is imperfect, e.g.\ the the number of ancestral populations that intermated to create the admixed individual is not always well known. \item[(3)] AIMs are not perfectly ancestry informative for all population distinctions.
\end{enumerate}
Theorem~\ref{T1} below gives a result on asymptotic normality of the ML-estimator when taking into account the effects of~(1), so it also leads to an estimate of the variance of $\hat q$. Moreover, we note that \citep{PfaffLong2004} have studied the variance of $\hat q$ due to~(3) -- not perfectly informative markers -- and we recall their ideas in Remark~\ref{rem:Minfty}. In both cases, we apply an approach involving a central limit theorem already discussed in \cite{Tang2005}, and which is based on the theory of asymptotic normality of ML-estimators; see e.g. Theorem 6.3.10 in \cite{Lehmann}. While \cite{Tang2005} is dealing with the task of simultaneously estimating $p$ and $q$, they conclude that the theory of asymptotic normality of the ML-estimator is hardly applicable since it requires inversion of a large matrix. However, when estimating $q$ separately, this inversion only involves a $K\times K$-matrix (see $\Sigma_x$ in Theorem~\ref{T1}) which is easily done since $K\leq 7$ in most applications. Consequently, in both cases, (1) and (3), we reach some asymptotic results which can be used efficiently in computation. We will  discuss measurement error due to (1)--(3) in the context of our simulations in Section~\ref{S32}. 

\noindent
{\color{black}We now re-derive -- using our notation -- results from \citep{PfaffLong2004}, collected in \eqref{eq:6} and \eqref{eq:Vqell2}, concerning variance of $\hat q$ in the case of many loci, i.e.\ large $M$.}

\begin{remark}[Extending the results from \citep{PfaffLong2004}: The case of large $M$]\label{rem:Minfty}
We will stick to our notation with empirical measures and $q$ with $\langle q, 1\rangle = 1$, and recall and extend results by \citealp{PfaffLong2004}. 
Assuming a large number $M$ of loci, we will now show that the covariance matrix of $\hat q$ is -- if we consider (3) from above as the only source of uncertainty -- approximately given by
\begin{align}\label{eq:6}
    \mathbb V_q[\hat q - q] & \approx \frac{1}{2M} (\Sigma - q\cdot q^\top), \qquad \Sigma^{-1} := \int \frac{\alpha \cdot \alpha^\top}{\langle q, \alpha\rangle} \mu_p(d\alpha)
\end{align}
with the empirical measure
\begin{align*}
  \mu_{p}(d\alpha) = \mu_{p,x}(d\alpha, \{0,1,2\}) = \frac 1{M}
  \sum_{m=1}^M\sum_{i=1}^I \delta_{p_{\cdot im}}(d\alpha) \in \mathcal M_{I}([0,1]^{K}).
\end{align*}
Note that $\Sigma^{-1} \cdot q = 1$ (since $\sum_i p_{kim} = 1$), and therefore $(\Sigma - q\cdot q^\top) \cdot 1 = q-q = 0$, i.e.\ $1$ is eigenvector of $\mathbb V_q[\hat q - q]$ for the eigenvalue~0. Moreover, since $\mathbb V_q[\hat q - q]$ is symmetric, this implies that the eigenvectors are orthogonal, and assuming that $\hat q - q$ is normal 
\begin{align}\label{eq:Vqell2}
\mathbb V_q[||\hat q - q||_2^2] \text{ is approximately the sum of all eigenvectors of }  \frac{1}{2M}(\Sigma - q\cdot q^\top).
\end{align}
\\
The claim \eqref{eq:6} is analogous to the statement of \citep{PfaffLong2004} that the information of a set of AIMs is proportional to the inverse of the Fisher information, $\mathbb E_q[\nabla^2\ell(p,q|X)]$. (However, note that \citealp{PfaffLong2004} use a different parametrization.) In order to see \eqref{eq:6}, we use the Taylor series expansion (and Lemma~\ref{l:1} for the first equality) in order to obtain
\begin{align} \label{eq:taylor0}
  1 & = \nabla \ell(p, \hat q|X) \approx \nabla \ell(p, q|X) + \nabla^2 \ell(p, q|X)(\hat q - q).
\end{align}
With the law of large numbers and $\mathbb E_q[X_{im}] = 2\langle q, p_{\cdot im}\rangle$, we obtain from \eqref{eq:nabla2ell}
\begin{align}\label{eq:Sigmax-Sigma}
  \nabla^2 \ell(p, q|X) & = -\Sigma_x^{-1} \approx - \int \frac{\alpha \cdot \alpha^\top}{\langle q, \alpha\rangle} \mu_p(d\alpha) = -\Sigma^{-1},
\end{align}
which is symmetric and strictly negative definite (hence invertible). Moreover, 
$$\mathbb E_q[\nabla \ell(p, q|X)] = \frac{1}{M} \sum_{m=1}^M \sum_{i=1}^I \langle q, p_{\cdot im}\rangle \frac{p_{\cdot im}}{\langle q, p_{\cdot im}\rangle} = 1$$ and, by using $\mathbb{COV}_q[X_{im}, X_{jm}] = 2\langle q, p_{\cdot  im}\rangle(\delta_{ij} - \langle q, p_{\cdot j m}\rangle)$ and independence of
markers, the matrix
\begin{align*}
  \mathbb{V}_q[\nabla \ell(p, q|X)] & = \frac{1}{4M^2} \sum_{m=1}^M \sum_{i,j=1}^I        \Big(\frac{p_{\cdot im}\cdot p_{\cdot jm}^\top }{\langle q, p_{\cdot im}\rangle\langle q, p_{\cdot jm}\rangle} \Big)2\langle q, p_{\cdot im} \rangle(\delta_{ij} - \langle q, p_{\cdot j m}\rangle)
  \\ & = \frac{1}{2M^2} \Big(\sum_{m=1}^M \Big(\sum_{i=1}^I \frac{p_{\cdot im} \cdot p_{\cdot im}^\top}{\langle q, p_{\cdot im}\rangle}\Big) - 1 \cdot 1^\top \Big) \\ & = \frac 1{2M}\Big( \int \frac{\alpha \cdot \alpha^\top}{\langle q, \alpha\rangle}  \mu_p(d\alpha) - 1\cdot 1^\top\Big) = \frac{1}{2M}(\Sigma^{-1} - 1\cdot 1^\top). 
\end{align*}
Finally, we obtain from \eqref{eq:taylor0} that, approximately, (using that $\mathbb V[A\cdot X] = A \cdot \mathbb V[X] \cdot A^\top$ for some random vector $X$ and a matrix $A$)
\begin{align*}
    \mathbb V_q[\hat q - q] & = \frac 1{2M} \Sigma \cdot (\Sigma^{-1} - 1\cdot 1^\top) \cdot \Sigma = \frac{1}{2M} (\Sigma - q\cdot q^\top),
\end{align*}
where the last equality follows from $\Sigma^{-1} \cdot q = 1$.

\begin{remark}[The bi-allelic case~2]\label{rem:bi2}
  In the bi-alleic case from Remark~\ref{rem:bi1}, we find the Hessian of $\ell$ (see \eqref{eq:nabla2ell})
  \begin{align}\label{eq:biSigmax}
      \Sigma_x^{-1} := -\nabla^2 \ell(p,q|x) = \tfrac 12 \int \Big(z \frac{\alpha\cdot \alpha^\top}{\langle q, \alpha\rangle^2} + (2-z)\frac{(1-\alpha) \cdot(1-\alpha)^\top}{\langle q, 1-\alpha\rangle^2}\Big) \mu_{p,x}(d\alpha, dz).
  \end{align}
  For large $M$, we can write after some simplifications 
  \begin{align}\label{eq:biSigma}
      \Sigma^{-1}_x \approx \Sigma^{-1} & :=  \int \Big(\frac{\alpha \cdot \alpha^\top}{\langle q, \alpha\rangle} + \frac{(1-\alpha) \cdot (1-\alpha)^\top}{1 - \langle q, \alpha\rangle}\Big) \mu_p(d\alpha) \\ & = (E_K - 1\cdot q^\top) \cdot \int \frac{\alpha \cdot \alpha^\top}{\langle q, \alpha\rangle\langle q, 1-\alpha\rangle} \mu_p(d\alpha)\cdot (E_K - q \cdot 1^\top) + 1\cdot 1^\top\notag
\end{align}
with (recall $\mu_{p,x}$ from \eqref{eq:bimeas1})
$\mu_p(d\alpha) = \mu_{p,x}(d\alpha, \{0,1,2\}) \in \mathcal M_1([0,1]^K)$.
\end{remark}

\end{remark}

\noindent
In our main Theorem below, we are treating the situation where the set of markers (of cardinality $M$) is fixed, but the frequencies $p$ come with some uncertainty since they are only estimated from a reference database of finite size $N$ (diploids). We are going to study the dependence of the ML-estimator $\hat q$ on the distribution of the frequencies $P$. Our main example is {\color{black} a reference database consisting of $N_k$ (diploid) individuals from population $k$, i.e.} $2N_k P_{k\cdot m} = (2N_k P_{kim})_{i=1,...,I} \sim \text{Mult}(2N_k, (p_{kim})_{i=1,...,I})$ for a fixed family $(p_{kim})_{k=1,...,K, i=1,...,I, m=1,...,M}$ and $(N_k)_{k=1,...,K}$ with $N_k \approx r_k N$ for some $r\in\mathbb S^\circ$, and $(2N_k P_{k\cdot m})_{k=1,...,K, m=1,...,M}$ are independent. (Here, $N = N_1 + \cdots + N_K$.) Note that $2N\cdot \mathbb{COV}[P_{kim}, P_{kjm}] = p_{kim}(\delta_{ij}-p_{kjm})/r_k$ in this case. We now formulate the (abstract/mathematical) result. {\color{black} Note that \eqref{eq:T1a} implies a formula for the variance of the estimator of $q$ similar to \eqref{eq:6} in the case of large $M$; see also Remark~\ref{rem:MNinfty} for more details.} Applications involving simulations and forensic genetic databases can be found in Section~\ref{S3}.

\begin{theorem}[Central Limit Theorem for large reference databases]
  \label{T1}
  Let $r\in \mathbb S^\circ, p \in \mathbb S_{I}^{K\times M}$,
  $x = (x_{im})_{i=1,...,I, m=1,...,M} \in \{0,1,2\}^{I\times M}$, and let Assumptions~\ref{ass0} and~\ref{ass1} hold. For the
  log-likelihood $\ell(.,.|x)$ for fixed $x$ from \eqref{eq:L|P}, let
  $P = P^N$ be random and such that $P^N \xrightarrow{N\to\infty} p$
  in probability, and
  $$\sqrt{2N}(P^N_{k\cdot m} - p_{k\cdot m}) \xRightarrow{N\to\infty} Z_{km} \sim N(0, \Pi(r_k,p_{k\cdot m}))) \text{ with }  (\Pi(r,p))_{ij} = r^{-1} \cdot p_{i}(\delta_{ij} - p_{j}),$$ and $(Z_{km})_{k=1,...,K, m=1,...,M}$ are independent. Moreover, let $\hat Q^N$ be the ML-estimator based on $P^N$, i.e.\ $\hat Q^N$ maximizes $q\mapsto \ell(P^N,q|x)$, and $\hat q \in \mathbb S^\circ$ the ML-estimator for infinite $N$, i.e.\ $\hat q$ is the unique maximizer of $q\mapsto \ell(p,q|x)$. Then, 
  \begin{align}\label{eq:T1a}
    \sqrt{4NM}(\hat Q^N - \hat q) \xRightarrow{N\to\infty} Z \sim N(0, \Sigma_x \cdot \Gamma_x
    \cdot \Sigma_x)
  \end{align}
  with $\Sigma_x^{-1}$ as in \eqref{eq:nabla2ell}, 
  \begin{align}
    \Gamma_{x} & := \frac 1{2M} \sum_{m=1}^M \sum_{i,j=1}^I \Delta(x_{im},p_{\cdot im})   \cdot \diag(r^{-1}) \cdot\diag(p_{\cdot im}(\delta_{ij} - p_{\cdot jm})) \cdot \Delta(x_{jm}, p_{\cdot jm})^\top,\label{eq:T1b}
    \\
    \Delta(p,x) & :=\notag   x\Big(\frac{E_K}{\langle \hat q, p\rangle} -
                \frac{p \cdot \hat q^\top}{\langle \hat q, p\rangle^2}\Big).
  \end{align}
\end{theorem}

\begin{remark}[The case of large $M$ and $N$]\label{rem:MNinfty}
  It is possible to combine the results from Remark~\ref{rem:Minfty} and Theorem~\ref{T1}. This then gives the approximation, since $\Sigma_x \xrightarrow{M\to\infty}\Sigma$,
  \begin{align}\label{eq:Vqall}      \mathbb V_q[\hat Q^N - q] & \approx \mathbb V_q[\hat Q^N - \hat q] + \mathbb V_q[\hat q - q] \approx \frac{1}{2M}\Big(\Sigma - q\cdot q^\top + \frac{1}{2N} \Sigma \cdot \Gamma \cdot \Sigma\Big) , \\ \notag
      \Gamma & := \frac{1}{2M} \sum_{m=1}^M \sum_{i,j=1}^I \Delta(p_{\cdot im}) \cdot \diag(r^{-1}) \cdot \diag(p_{\cdot im}(\delta_{ij} - p_{\cdot jm})) \cdot \Delta(p_{\cdot jm}), \\ \Delta(p) & := 2\Big(E_K - \frac{p\cdot q^\top}{\langle q, p\rangle}\Big), \notag     
  \end{align}
  and all matrices are independent of $x$.
  \\
  It has been noted previously that estimates of the variance of $q$ 
  can be used in order to estimate the number of markers needed in order to reduce this variance below some threshold. This was established by \cite{Rosenberg2003} -- based on results by \cite{Millar1991} -- by assuming infinite $N$ and that the allele frequency differential at all loci has a lower bound. Similar thoughts can be applied in the setting of \eqref{eq:Vqall}, leading to lower bounds of $M, N$ in order to reduce the $\max_q \mathbb V_q[\hat Q^N - q]$ below some threshold.
\end{remark}

\begin{remark}[$Z$ is degenerate]\label{rem:deg}
  We note that $1^\top \cdot (\hat Q^N - \hat q) = 0$ since both,
  $\hat Q^N$ and $\hat q$ sum to~1. With $Z$ as in \eqref{eq:T1a}, we
  now show that $1^\top\cdot Z = 0$ as well, i.e.\ $Z$ from
  \eqref{eq:T1a} has a degenerate normal distribution.
  \\
  Indeed: $1^\top Z$ has covariance
  $1^\top \Sigma_x\cdot \Gamma_x\cdot\Sigma_x \cdot 1$. Since $\hat q$ is the ML-estimator, we find that  $\Sigma_x^{-1} \cdot \hat q = \nabla \ell(p,q|x) = 1$ (see Lemma~\ref{l:1}) which implies $\Sigma_x \cdot 1 = \hat q$. Since for $m=1,...,M$,
  \begin{align*}
    \Delta(x,p)^\top \cdot \hat q & = x\Big(\frac{\hat q^\top }{\langle \hat q,p\rangle} - \frac{\hat q^\top \cdot p \cdot \hat q^\top}{\langle \hat q,p\rangle^2}\Big) = 0,
  \end{align*}
  we find that $1^\top \cdot \Sigma_x \cdot \Gamma_x \cdot \Sigma_x \cdot 1 = 0$. This implies that $1^\top \cdot Z = 0$.
\end{remark}

\begin{proof}[Proof of Theorem~1]
  First, we write (recall from \eqref{eq:L|P})
  \begin{align*}
    \ell(P^N,q|x) = C_x + \int z\log(\langle q,\alpha\rangle)\mu_{P^N,x}(d\alpha, dz).
  \end{align*}
  with $C_x$ not depending on $P^N$ and $q$. From
  $P^N \xrightarrow{N\to\infty} p$ in probability, we see that
  $\mu_{P^N,x} \xRightarrow{N\to\infty} \mu_{p,x}$ and therefore
  \begin{equation}
    \label{eq:T15}
    \begin{aligned}
      \nabla\ell(P^N,\hat q|x) & = \int z\frac{\alpha}{\langle \hat q,
        \alpha\rangle} \mu_{P^N,x}(d\alpha, dz)
      \xRightarrow{N\to\infty} \nabla\ell(p, \hat q|x) = 1, \\
      \nabla^2\ell(P^N,\hat q|x) & = -\int z\frac{\alpha \cdot
        \alpha^\top}{\langle \hat q, \alpha\rangle^2}
      \mu_{P^N,x}(d\alpha,dz) \xRightarrow{N\to\infty} -\Sigma_x^{-1}.
    \end{aligned}
  \end{equation}
  We need to make the first convergence more precise and write, using
  a Taylor approximation,
  \begin{equation}
    \label{eq:T16}
    \begin{aligned}
      \sqrt{2N}(\nabla& \ell(P^N,\hat q|x) - 1) = \sqrt{2N}\Big(\Big(\frac 1{2M} \sum_{m=1}^M \sum_{i=1}^I x_{im} \frac{P^N_{\cdot im}}{\langle \hat q, P^N_{\cdot im}\rangle} \Big)-1\Big) \\ & = \frac 1{2M} \sum_{m=1}^M \sum_{i=1}^I x_{im} \Big(\frac{ E_K}{\langle \hat q, p_{\cdot im}\rangle} - \frac{p_{\cdot im} \cdot \hat q^\top}{\langle \hat q, p_{\cdot im}\rangle^2}\Big) \cdot \sqrt{2N}(P^N_{\cdot im} - p_{\cdot
        im}) + o_P(1),
    \end{aligned}
  \end{equation}
  where $o_P(1)$ is a sequence converging to zero in probability. We
  obtain that for $N\to\infty$, the right-hand-side converges weakly to a normally distributed random variable $Y$ with covariance matrix $\Gamma_x$.
Now, using that $\hat Q^N$ is the ML-estimator based on $P^N$, again using a Taylor approximation,
\begin{align*}
    1 & = \nabla\ell(P^N,\hat Q^N|x)
        = \nabla \ell(P^N,\hat q|x) + \nabla^2 \ell(P^N,\hat q|x) \cdot (\hat Q^N - \hat q)
        + o_P(1/\sqrt{N}).
  \end{align*}
  Using the convergence of $\nabla^2\ell(\hat q, P^N|x)$ from
  \eqref{eq:T15}, and the convergence from \eqref{eq:T16}, we find,
  solving the last equation for $(\hat Q^N - \hat q)$ and multiplying
  with $\sqrt{2N}$,
  \begin{align*}
    \sqrt{2N} (\hat Q^N - \hat q) & \xRightarrow{N\to\infty} \Sigma_x \cdot Y.
  \end{align*}
  Now, the result follows since $\Sigma_x$ is symmetric and hence
  $\Sigma_x \cdot Y \sim N(0, \Sigma_x \cdot \Gamma_x \cdot \Sigma_x)$.
\end{proof}

\begin{remark}[The bi-allelic case~3]\label{rem:bi3}
  For Theorem~\ref{T1}, the bi-allelic case in fact allows for some more simplifications. For notation, see Remarks~\ref{rem:bi1} and~\ref{rem:bi2}. Here, we have $r\in \mathbb S^\circ$ as in Theorem~\ref{T1} and $\ell$ as in \eqref{eq:ellbi}. Let $P = P^N$ be random and such that $P^N \xRightarrow{N\to\infty} p$, and  $$\sqrt{2N}(P^N_{km} - p_{km}) \xRightarrow{N\to\infty} Z_{km}  \sim N(0, \Pi(r_k,p_{km}))) \text{ with } \Pi(r,p) = r^{-1} \cdot p(1 - p),$$ and $(Z_{km})_{k=1,...,K, m=1,...,M}$ are independent. As above, let $\hat Q^N$ be the ML-estimator based on $P^N$, and $\hat q \in \mathbb S^\circ$ the ML-estimator for infinite $N$. Then, \eqref{eq:T1a} holds with $\Sigma_x$ from \eqref{eq:biSigmax} and
  \begin{align*}
    \Sigma_x^{-1} & := \frac 12 \int \Big(
    z \frac{\alpha \cdot \alpha^\top}{\langle \hat q, \alpha\rangle^2} + (2-z) \frac{(1-\alpha) \cdot (1-\alpha)^\top}{\langle \hat q, 1-\alpha\rangle^2}
    \Big)\mu_{p,x}(d\alpha,dz),
    \\ \Gamma_x & :=\frac{1}{4M} \int \Big( \widetilde\Delta(z, \alpha) \cdot \diag(r^{-1}) \cdot \diag(p(1-p)) \cdot \widetilde \Delta(z, \alpha)^\top \Big)\mu_{p,x}(d\alpha, dz), \\ \widetilde\Delta(x,p) & := \Delta(x,p) - \Delta(2-x, 1-p) = x\Big(\frac{E_K}{\langle \hat q, p\rangle} - \frac{p \cdot \hat q^\top}{\langle \hat q, p\rangle^2}\Big) - (2-x)\Big(\frac{E_K}{\langle \hat q, 1-p\rangle} - \frac{(1-p) \cdot \hat q^\top}{\langle \hat q, 1-p\rangle^2}\Big).
  \end{align*}
  The form of $\Sigma_x$ here follows directly from \eqref{eq:T1a} and \eqref{eq:bimeas1}. For $\Gamma_x$ from \eqref{eq:T1b}, we write by evaluating the sums over $i$ and $j$, which leads -- using $\Pi(r,p) := \diag(r^{-1}) \cdot \diag(p(1-p))$ -- from \eqref{eq:T1b} to
  \begin{align*}
    \Gamma_x & = \frac{1}{4M^2} \sum_{m=1}^M  \Delta(x_m, p_{\cdot m})\cdot \Pi(r,p_{\cdot m}) \cdot \Delta(x_m, p_{\cdot m})^\top \\ & \qquad \qquad \qquad +  \Delta(2-x_m, 1-p_{\cdot m})\cdot \Pi(r,p_{\cdot m}) \cdot \Delta(2-x_m, 1-p_{\cdot m})^\top \\ & \qquad \qquad \qquad - \Delta(x_m, p_{\cdot m}) \cdot \Pi(r,p_{\cdot m})\cdot \Delta(2-x_m, 1-p_{\cdot m})^\top \\ & \qquad \qquad \qquad - \Delta(2-x_m, 1-p_{\cdot m}) \cdot \Pi(r,p_{\cdot m}) \cdot \Delta(x_m, p_{\cdot m})^\top
    \\ & = \frac{1}{4M^2} \sum_{m=1}^M ( \Delta(x_m, p_{\cdot m}) - \Delta(2-x_m, 1-p_{\cdot m})) \cdot \Pi(r, p_{\cdot m}) \cdot ( \Delta(x_m, p_{\cdot m}) - \Delta(2-x_m, 1-p_{\cdot m}))^\top  
\end{align*}
and the result follows from the definitions of $\widetilde\Delta(x,p)$ and $\Pi(r,p)$.

This setting also allows to combine the limits for $M\to\infty$ and $N\to\infty$. We obtain \eqref{eq:Vqall} with $\Sigma$ from \eqref{eq:biSigma} and
\begin{align*}
\Gamma & := \frac 1{4M} \int \Big(\widetilde\Delta(\alpha) \cdot \diag(r^{-1}) \cdot \diag(p(1-p))\cdot \widetilde\Delta(\alpha)^\top\Big) \mu_p(d\alpha), \\ \widetilde\Delta (p) & = 2\Big( \frac{(1-p)\cdot \hat q^\top}{\langle \hat q, (1-p)\rangle} - \frac{p \cdot \hat q^\top}{\langle \hat q, p\rangle} \Big) = 2\frac{1 \cdot \hat q \cdot ^\top \cdot p \cdot \hat q^\top - p\cdot \hat q^\top}{\langle \hat q, p\rangle\langle \hat q, 1-p\rangle}
\end{align*}
\end{remark}

\subsection{The ML-estimator as a stable fixed point}
\noindent
Computing the ML-estimator $\hat q$, i.e.\ solving
$\nabla \ell(p,q|x) = 1$ for $q$ is possible by Newton-Raphson
iteration. However, care must be taken in order to ensure that
$\langle \hat q,1\rangle = 1$ and $q_k\geq 0$ for all $k=1,...,K$. We
have found a different iterative way of finding $\hat q$ which
guarantees that $\hat q \in \mathbb S$. It is based on finding a
stable fixed point of some $\mathbb S$-valued function
$q\mapsto L(p,q|x)$.

\begin{theorem}[The ML-estimator as a stable fixed point]\label{T2}
  Let $p \in \mathbb S_{I}^{K\times M}$,
  $x = (x_{im})_{i=1,...,I, m=1,...,M} \in \{0,1,2\}^{I\times M}$, and
  let Assumptions~\ref{ass0} and~\ref{ass1} hold (i.e., $\ell$ is
  strictly concave). Let $\ell(p,.|x)$ be as above and define
  \begin{align*}
    &L(p,.|x): \begin{cases} \mathbb S & \to \mathbb S \\ q & \mapsto
      \diag(q) \cdot \nabla \ell(p,q|x) = \diag(q) \cdot
      \displaystyle \int z\frac{\alpha}{\langle q, \alpha\rangle}
      \mu_{p,x}(d\alpha, dz).\end{cases}
  \end{align*}
  Then, for $q^\ast \in \mathbb S^\circ$, the following statements are equivalent:
  \begin{enumerate}
  \item $q^\ast$ is the unique and global maximum of $\ell(p,.|x)$;
  \item $q^\ast$ is a local maximum of $\ell(p,.|x)$;
  \item $q^\ast$ is a locally stable fixed point
    of $L(p,.|x)$.
  \item $q^\ast$ is the unique locally stable fixed point of $L(p,.|x)$.
  \end{enumerate}
\end{theorem}

\begin{remark}[Computing $\hat q$ using $L$]
  Given that Assumption \ref{ass1} holds, the above result together
  with Lemma~\ref{l:1} implies that $q^\ast \in \mathbb S^\circ$ is
  ML-estimator if and only if $L(p,q^\ast|x) = q^\ast$. Hence, if the
  iteration $q_{n+1} := L(p,q_n|x)$, converges, we are sure to have
  found the ML-estimator. In our numerical applications in
  Section~\ref{S3}, we will compute ML-estimators in this way.
  \\
  One word of caution is necessary if
  $q\in \partial \mathbb S = \mathbb S \setminus \mathbb S^\circ$, the
  edge of $\mathbb S$. In this case, Lemma~\ref{l:1} and  Theorem~\ref{T2} are inconclusive, but our numerical results suggest  ML-estimators are still stable fixed points of $L$. Moreover, note that $\partial \mathbb S$ is the union of lower-dimensional  simplices, and often we can apply Lemma~\ref{l:1} and  Theorem~\ref{T2} on these lower-dimensional manifolds.
\end{remark}

\begin{proof}[Proof of Theorem~\ref{T2}]
  Clearly, $1.\Rightarrow 2.$ and $4.\Rightarrow 3.$ are immediate. By
  strict concavity of $\ell$, we also find $2.\Rightarrow 1.$ Assuming
  that $2.\iff 3.$ is established, assume that $4.$ does not hold, but
  $3.$ Then, there is a second locally stable fixed point
  of $L(p,.|x)$. This would (by $3.\iff 2.\iff 1.$) be
  another global maximum of $\ell$, which is a contradiction. Hence,
  $3.\Rightarrow 4.$ follows, and we are left with showing $2.\iff 3.$
  For this, we will use the fact shown in Lemma~\ref{l:1} that
  $q^\ast$ maximizes $\ell(p,.|x)$ if and only if
  $\nabla\ell(p,q|x)=1$. Recall $\nabla\ell(p,.|x)$ and
  $\nabla^2\ell(p,.|x)$ from \eqref{eq:nablaell} and
  \eqref{eq:nabla2ell}, we start by noting that
    \begin{equation}
      \label{eq:anabla2}
      \begin{aligned}
        \nabla L(p,q|x) & = \diag(\nabla \ell(p,q|x)) + \diag(q) \cdot
        \nabla^2\ell(p,q|x).
      \end{aligned}
    \end{equation}
    Since $\ell$ is strictly concave, for all $q \in \mathbb S$, all
    eigenvalues of $\nabla^2\ell(p,q|x)$ are
    negative. \\
    $2.\Rightarrow 3.$ From the definition of $L$ and Lemma~\ref{l:1},
    we already see that if $q^\ast \in \mathbb S^\circ$ is
    global/local maximum of $\ell(p,.|x)$, then
    $\nabla\ell(p,q^\ast|x) = 1$, hence
    $L(p,q^\ast|x) = \diag(q^\ast) \cdot 1 = q^\ast$, i.e.\ $q^\ast$
    is a fixed point of $L(p,.|x)$. It remains to show local
    stability.  According to the Stable-Manifold-Theorem for discrete
    dynamical systems (see e.g.\ Theorem 4.7 of \cite{Galor2007}), we
    must show that all eigenvalues of $\nabla L(p,q^\ast|x)$ have
    absolute value $<1$.  We have from \eqref{eq:anabla2}
    $$\nabla L(p,q^\ast|x) = E_K -B, \qquad B = - \diag(q^\ast) \cdot \nabla^2 \ell(p,q^\ast|x)
    = \int z \frac{\diag(q^\ast) \cdot \alpha \cdot
      \alpha^\top}{\langle q^\ast, \alpha\rangle^2} \mu_{p,x}(d\alpha,
    dz).$$ We note that all entries of $B$ are non-negative and
    \begin{align*}
      1^\top \cdot B & =\int z \frac{(q^\ast)^\top \cdot \alpha \cdot 
                       \alpha^\top}{\langle q^\ast, \alpha\rangle^2}\mu_{p,x}(d\alpha, dz)
                       = \int z \frac{\alpha^\top}{\langle q^\ast, \alpha\rangle}\mu_{p,x}(d\alpha, dz)
                       = \nabla\ell(p,q^\ast|x)^\top
                       = 1^\top,
    \end{align*}
    i.e.\ $B$ is the transpose of a stochastic matrix. Moreover,
    $$B\cdot q^\ast = \int z \frac{\diag(q^\ast) \cdot \alpha \cdot \alpha^\top \cdot q^\ast}
    {\langle q^\ast,\alpha\rangle^2} \mu_{p,x}(d\alpha,dz) =
    \diag(q^\ast) \cdot \nabla\ell(p,q^\ast) = \diag(q^\ast) \cdot 1 =
    q^\ast,$$ i.e.\ $q^\ast$ is a right-eigenvector for the
    eigenvalue~1. We also see that the Markov chain with transition
    matrix $B^\top$ is irreducible and aperiodic, hence there are (up
    to constant factors) unique left- and right-eigenvectors of $B$ to
    the eigenvalue~1. Also, note that the eigenvalues of $B$ and of
    $$ \diag(\sqrt{q^\ast}) \cdot \Big(\int z \frac{\alpha\cdot \alpha^\top}{\langle q^\ast, \alpha\rangle^2}
    \mu_{p,x}(d\alpha, dz)\Big) \cdot \diag(\sqrt{q^\ast})$$ coincide (see e.g.\ \citealp{HornJohnson2012}, Theorem 1.3.22),
    the latter being symmetric and positive definite (since
    $q^\ast \in \mathbb S^\circ$ and strict concavity of
    $\ell(p,.|x)$). All eigenvalues of such a matrix are real and in
    $(0,\infty)$. Since $B^\top$ is stochastic, this implies --
    according to Perron-Frobenius-theory -- that all eigenvalues of
    $B$ are in $(0,1]$. Hence, all eigenvectors of
    $\nabla L(p,q^\ast|x) = E_K-B$ are in $[0,1)$, and
    $2.\Rightarrow 3.$ follows.

    \noindent
    $3.\Rightarrow 2.$: If $q^\ast \in \mathbb S^\circ$ is a fixed
    point of $L$, we find that, for any $k=1,...,K$,
    $$ q_k^\ast = L(q^\ast)_k =
    \Big(\diag(q^\ast) \cdot \int z \frac{\alpha}{\langle q^\ast,
      \alpha\rangle} \mu_{p,x}(d\alpha, dz)\Big)_k = q_k^\ast \int z
    \frac{\alpha_k}{\langle q^\ast, \alpha\rangle} \mu_{p,x}(d\alpha,
    dz) .$$ Since $q_k^\ast > 0$, this implies
    $\int z \frac{\alpha_k}{\langle q^\ast, \alpha\rangle}
    \mu_{p,x}(d\alpha,dz) = 1$, i.e.\ $\nabla \ell(p,q^\ast|x) =
    1$. This shows 2.
\end{proof}

\begin{remark}[The bi-allelic case~4]\label{rem:bi4}
  For bi-allelic markers, we note that -- using the notation from \eqref{eq:bimeas1}  \begin{equation}
  \label{eq:ellbi4}
    \begin{aligned} 
    L(p,q|x) & = \frac 12 \diag(q) \cdot  \displaystyle \int \Big(z\frac{\alpha}{\langle q, \alpha\rangle} + (2-z)\frac{1-\alpha}{\langle q, 1-\alpha\rangle} \Big)  \mu_{p,x}(d\alpha, dz).
  \end{aligned}
  \end{equation}
\end{remark}

\section{Applications}
\label{S3}
We applied our results to both, simulated data and real data from a database used in forensic genetics. We implemented the estimators of the variance from Remark~\ref{rem:Minfty}, Theorem~\ref{T1} and Remark~\ref{rem:MNinfty}, and rely on finding the ML-estimator by using the iterative scheme based on Theorem~\ref{T2}. We focus on an implementation of the bi-allelic case only (see Remarks~\ref{rem:bi1}, \ref{rem:bi2}, \ref{rem:bi3}, \ref{rem:bi4}). Our implementation of all methods uses {\tt R} and can be downloaded from \url{http://github.com/pfaffelh/MNinfty}. Before we come to our results, we need some graphical tool.

\subsection{Displaying the variance {\color{black} in barplots}}\label{ss:31}
Estimators for the individual admixture $q$, based on genetic data $x$ from a single individual, are usually shown using a barplot, where each population comes with its own color, and sizes of bars in the same color are proportional to the estimated ancestry proportion from this population. In order to obtain a graphical representation of the uncertainty in the estimation of $q$  {\color{black} (i.e.\ the variance of the ML-estimator $\hat q$)}, we add to this barplot the estimators in variance from our theoretical results; {\color{black}see Figure~\ref{fig1}. A confidence region for $q$ can be computed from the covariance matrix, as e.g.\ given by the right hand side of \eqref{eq:Vqall}, denoted by $\Lambda$ in the sequel; for some more details see Remark~\ref{rem:kl}. If one is interested in the variance of $\hat q_k$ for some $k$, i.e.\ the individual admixture from one specific reference population, the answer is $\Lambda_{kk}$. However, it is as well interesting to understand covariances between $\hat q_k$ and $\hat q_\ell$, since they provide information on confounding effects between pairs of populations. In order to display the complete picture for variances and covariances, we use the orthogonal eigenvectors of $\Lambda$ which span the confidence region for the estimation of $q$, where we order the corresponding eigenvalues by their corresponding eigenvalues. The resulting $2(K-1)$ corners of this region can be displayed using error bars on top of a barplot; see the left part of Figure~\ref{fig1}. From these error bars, it is as well possible to read off the variation of single $\hat q_k$s. We note, however, that the true distribution of the ML-estimator is not normal, and even not symmetric around $q$, but follows a different shape. We nevertheless plot symmetric error bars since the true distribution is not accessible. Only for $K=3$, it is possible to display the three components of $q$ in a triangle, as well as their uncertainty; see the right part of Figure~\ref{fig1}. Here, eigenvectors of $\Lambda$ with lengths proprtional to the corresponding standard deviation, can be plotted directly in the plane.} Since the triangle plot only works for $K=3$, we will stick with barplots with error bars in the following.

\begin{remark}[More details on the confidence region]\label{rem:kl}
{\color{black}Let us add more detail how to obtain the error bars in a barplot.} Denote the right hand side in \eqref{eq:Vqall} by $\Lambda$, which is a symmetric $K\times K$ matrix. Hence, $\Lambda$ has orthogonal eigenvectors (with real, non-negative eigenvalues), and we know that $\Lambda \cdot 1 = 0$ (see Remarks~\ref{rem:Minfty} and~\ref{rem:deg}). Now, let $v$ be a unit-length eigenvector of $\Lambda$ for the eigenvalue $\lambda > 0$. By orthogonality, we find $\langle v, 1\rangle = 0$, which implies that $q + hv \in \mathbb S^\circ$ if $q\in\mathbb S^\circ$ and $h$ is small enough. In addition, assuming that the ML-estimator $\hat Q^N$ is normally distributed, we conclude that $v\cdot (\hat Q^N - q)\cdot v$ is normally distributed with mean~$q$ and variance $v^\top \cdot \Lambda \cdot v = \lambda$. From this, we see that $\mathbb P(|v\cdot (\hat Q^N - q)| > 2\sqrt\lambda) \approx 5\%,$ and we are able to find a confidence region (a subset of $\mathbb S$) for the estimation of $q$, when iterating over all eigenvectors. For the error bars, on top of the point estimate $\hat q$ of $q$, we display $\hat q \pm 2\sqrt{\lambda_i} v_i$, $i=1,...,K-1$ (but not exceeding $0$), where {\color{black}$v_i$} is the eigenvector of $\Lambda$ for the eigenvalue $\lambda_i, i=1,...,K-1$, starting with the largest eigenvalue. 
\end{remark}

\begin{figure}
\begin{center}
  \includegraphics[width=0.75\textwidth]{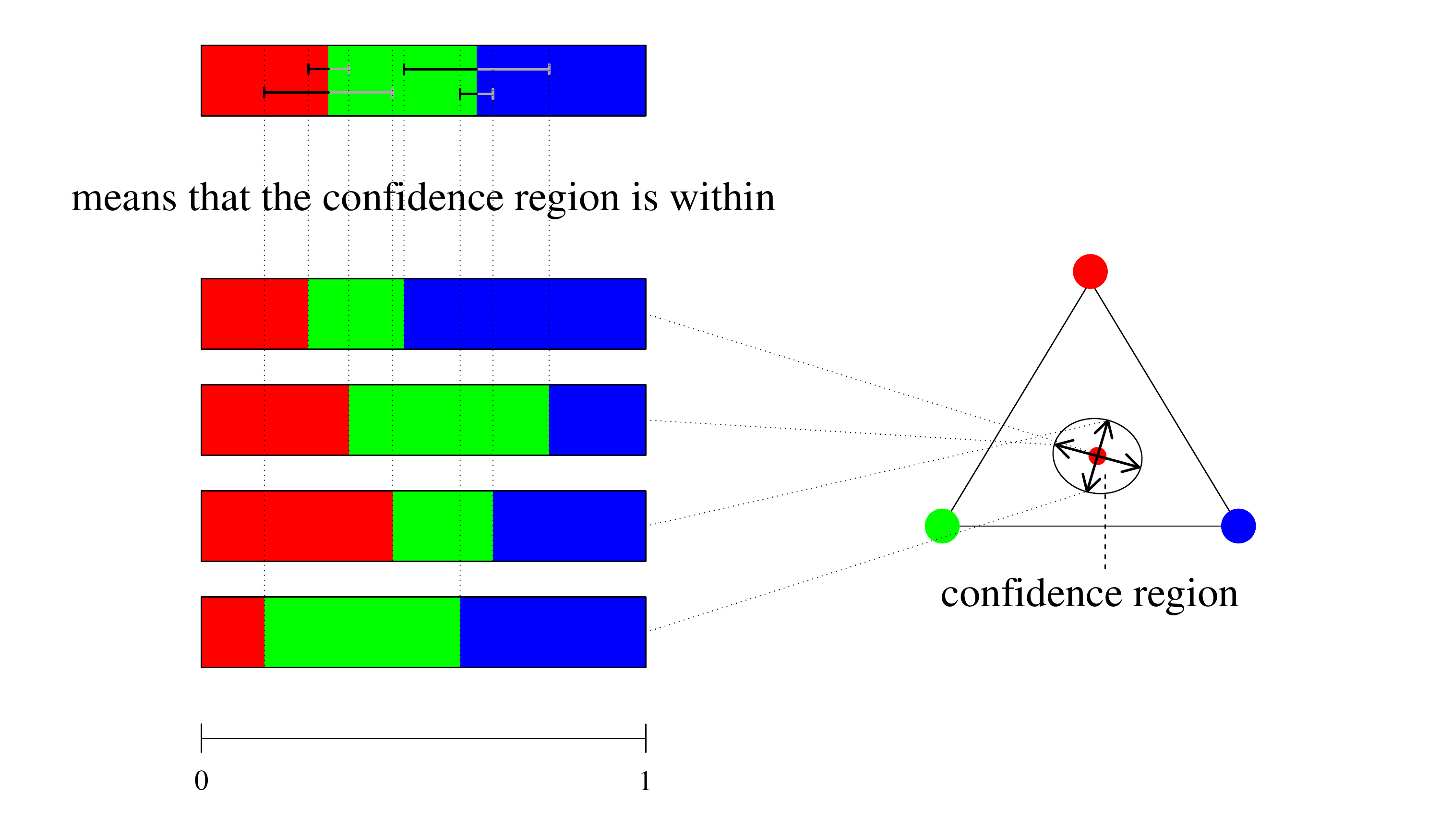}
  \caption{{\color{black} An illustration of the uncertainty in estimation of $q$ if there are $K=3$ ancestral populations.} On the left, the top row will be used in the sequel in order to display the uncertainty in the estimation of $q$ using a standard structure barplot. {\color{black} Here, we display $\hat q$ as usual using different colors in the barplot. On top of that, the error bars end at $\hat q \pm 2\sqrt{\lambda_i}v_i$, where $\lambda_i$ is the eigenvalue for the eigenvector $v_i$ of the covariance matrix of $\hat q$, $i=1,...,K-1$; see Remark~\ref{rem:kl}. For $i=1$  and $i=2$, we use the upper and lower error bars, respectively, while the distinction between the $\pm$-parts are displayed using the black and grey error bars in the top row. In other words, the corners of the confidence region are as displayed in the lower four rows. On the right, t}he same uncertaintly is displayed, but in a triangle plot, which is only available for $K=3$. The corners of the confidence region in the triangular plot match exactly the corresponding barplots on the left. The true $q$ is $(0.4, 0.3, 0.3)${\color{black}, and the ML estimate ist $\hat q \approx (0.29, 0.33, 0.38)$.}}
  \label{fig1}    
\end{center}
\end{figure}

\subsection{Simulations}
\label{S32}Recall the three sources of uncertaintly in estimating $q$ from Section~\ref{S23}. In our simulations, we start with (3) -- AIMs are not perfectly ancestry informative for all population distinctions -- and (1) -- The set of AIMs has imperfectly known ancestral population allele frequencies. Here, we simulate $M$ bi-allelic markers with independent $\beta(5,5)$ distributed allele frequencies in $K=3$ populations. (For such markers, the $F_{ST}$ between two populations is $0.047$ on average.) For the reference database, we choose $N$ (diploids) based on these allele frequencies, and estimate $q$ for an admixed sample. For estimates of the variance in $q$, there are two possibilities. Either, we use the result from Remark~\ref{rem:MNinfty}, or we use bootstrapping. For the latter, we resample our reference database, i.e.\ we choose $N$ individuals and $M$ markers with replacement. Doing this $B$ times gives our bootstrap samples, and each such sample gives rise to a new $\hat q$ for the test sample. This leads to $B$ estimates of $q$, and we can evaluate that empirical variances and co-variances. We see in Figure~\ref{fig2} that the variances from Theorem~\ref{T1} (or Remark~\ref{rem:MNinfty}) and the bootstrap estimates match {\color{black} well, at least qualitatively}. In addition, we see from this figure that increasing $M$ is much more efficient in reducing the uncertainty in the estimate of $q$ than increasing $N$: As can also be see from \eqref{eq:Vqall}, for a test sample with $q = (0.4, 0.3, 0.3)$, the variance in estimates of $q$ decreases linearly in $M$, but level off for large $N$ if $M$ is still moderate. We elaborate on this finding more deeply in Figure~\ref{fig7}. Here, for $K=3$ and $q=(0.4, 0.3, 0.3)$ as above, we can see the effects of finite $M$ and $N$ on the variance of $\hat q$ as in \eqref{eq:Vqell2}. In Figure~\ref{fig7}(A), we see that a larger number of markers $M$ leads to significant lower variance. The contribution of finite $N$ to the total variance decreases for large $N$, as shown in Figure~\ref{fig7}(B). 

\begin{figure}
  \includegraphics[width=0.95\textwidth]{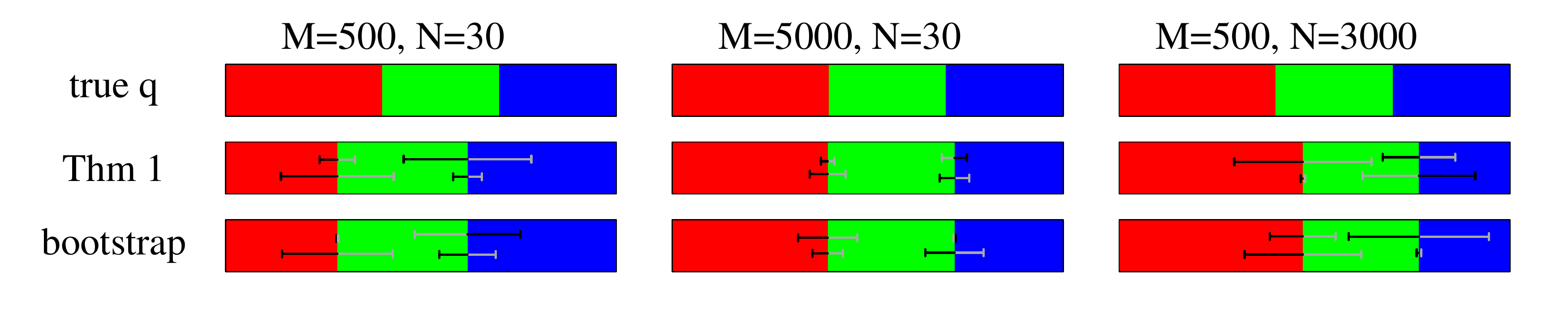}
  \caption{For $K=3$ ancestral populations and a single admixed test sample with $q = (0.4, 0.3, 0.3)$, we provide estimates of $q$ using a different number of markers $M$ and different sizes $N$ of the reference database. All allele frequencies are $\beta(5,5)$-distributed, {\color{black}hence have mean $0.5$ and standard deviation $\approx 0.15$}.}
  \label{fig2}
\end{figure}

\begin{figure}
\hspace{3cm}(A) \hspace{7.5cm}(B)
\begin{center}
\includegraphics[width=0.43\textwidth]{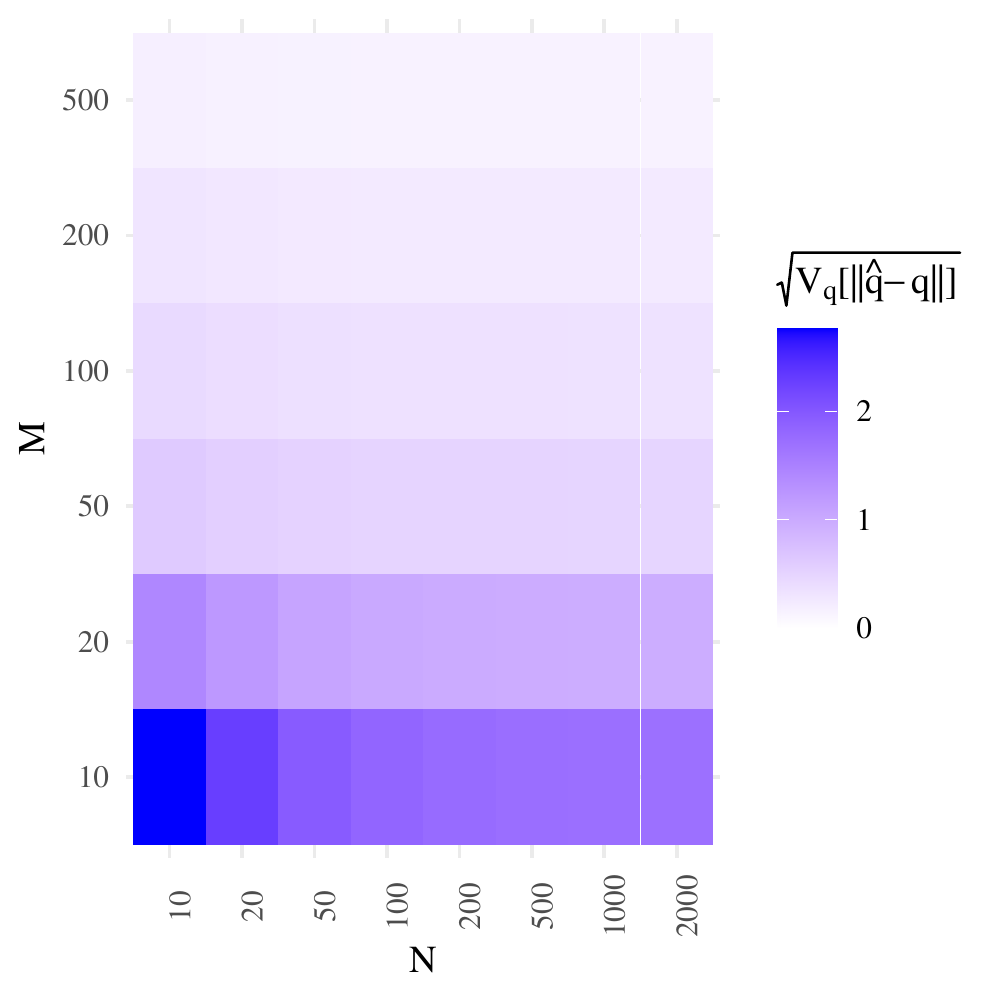}
\includegraphics[width=0.53\textwidth]{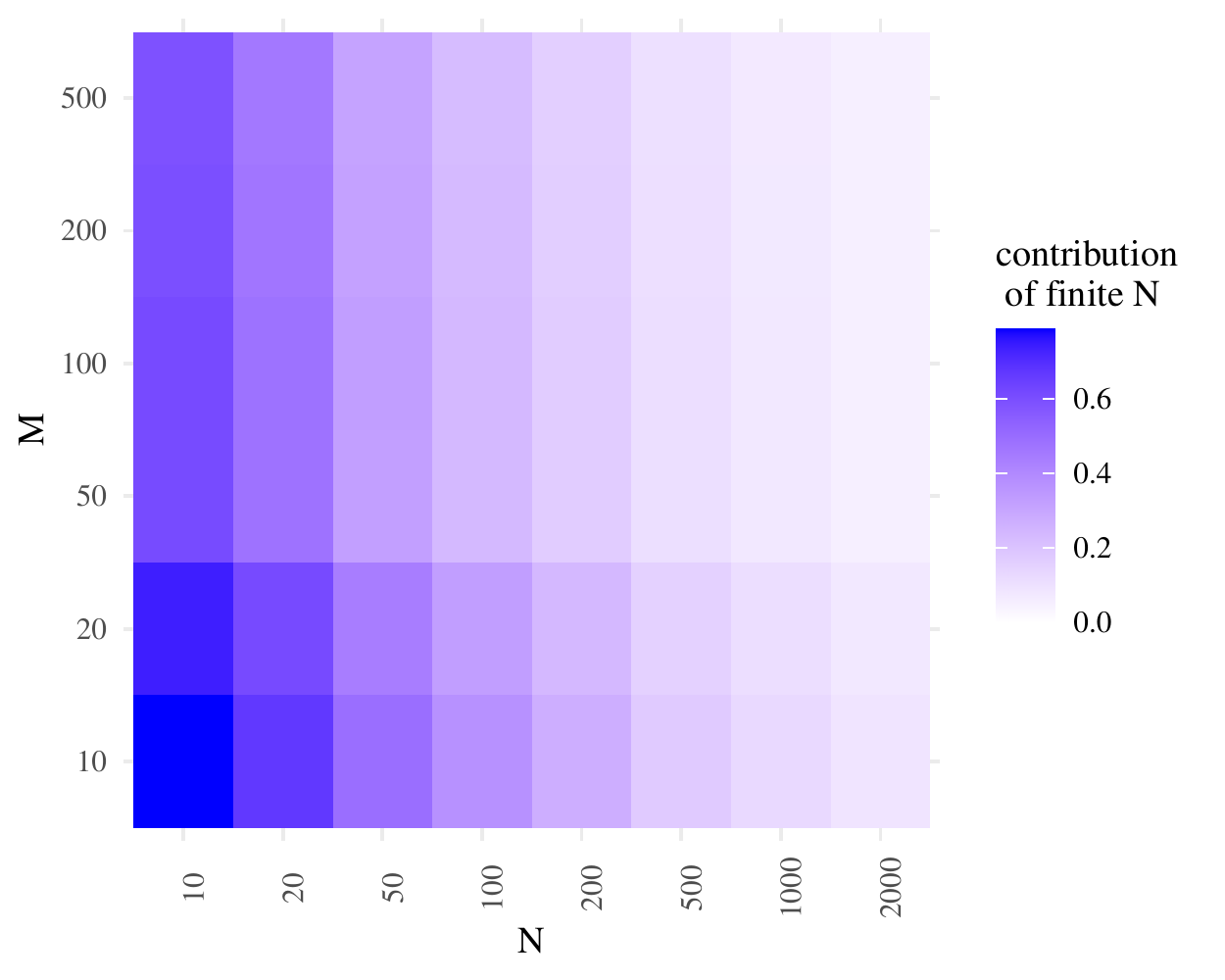}
\end{center}
\caption{Effects of finiteness of $M$ and $N$ on the variance {\color{black} of $\hat q$}. (A) Estimated standard deviation, measured as an $L^2$-distance between $\hat q$ and $q$, as in \eqref{eq:Vqell2}. True $q$ equals $(0.4, 0.3, 0.3)$ for three populations, where all alleles have $\beta(5,5)$-distributed frequencies.}
  \label{fig7}
\end{figure}

For (2) -- uncertainty in the ancestral populations contributing to the test sample -- we did two kinds of analysis. First, we extended the number of populations in the reference database up to $K=10$, although the test sample is admixed only from~2 populations. We see in Figure~\ref{fig3} that the overall picture of the two contributing populations also applies for large $K$, but some fraction is estimated to come from other populations, but also comes with a large variance. Second, we simulated the scenario of a sample originating outside of the reference database; see Figure~\ref{fig4}. So, the reference database has $K=3$ populations, but the test sample originates from a fourth population; see the middle column in Figure~\ref{fig4}. Note that the results (point estimators and their variances) can hardly be distinguished from results of individuals admixed from all three populations in the reference database (left column in Figure~\ref{fig4}). However, we note that the log-likelihoods for the individuals outside of the reference population at $\hat q$ (see \eqref{eq:ellbi}) is on average 90,7 units below the truely admixed individuals. For individuals admixed from one population inside and outside of the reference database (right column in Figure~\ref{fig4}), we still find an average difference of 28,7. 

\begin{figure}
\begin{center}
  \includegraphics[width=0.45\textwidth]{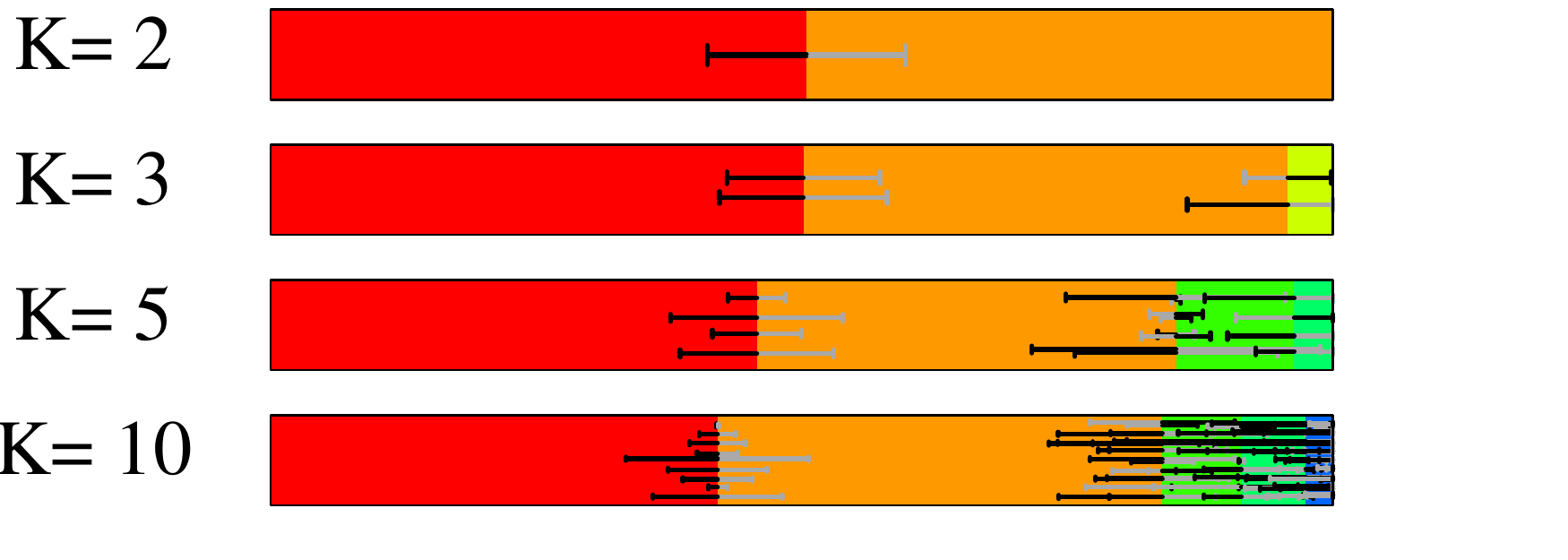}
\end{center}
\caption{For the same test sample {\color{black} as in Figure~\ref{fig7}}, we provide estimates of $q$ and their variances when the number of populations in the reference database increases. We use $M=500$ markers, distributed according to $\beta(5,5)$, and a reference database with 10 diploids per population.}
  \label{fig3}
\end{figure}

\begin{figure}
  \includegraphics[width=0.95\textwidth]{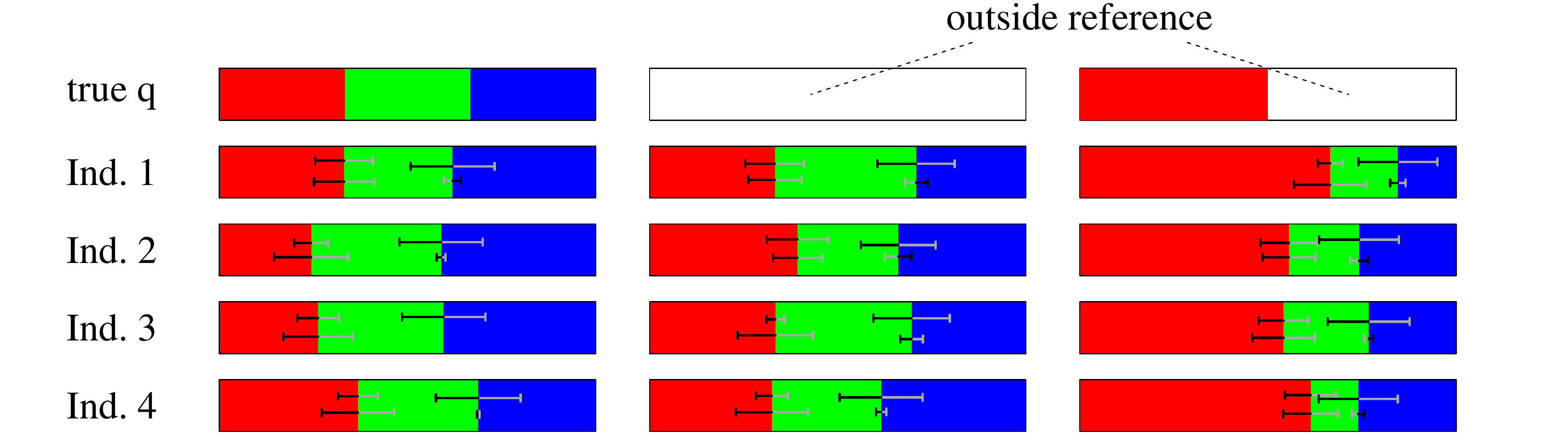}
  \caption{A comparison between estimates of $q$ if $q=c(1/3, 1/3, 1/3)$ (left), the sample is from a population outside the reference database (middle), and the sample is an admixture from one population within and outside the reference database. We use $M=500, N=300$. }
  \label{fig4}
\end{figure}

\subsection{Real-world data}
\label{S33}
As real-world examples, we consider a marker set used in forensic genetics, with an available reference (training) database available on the homepage of the software {\sc Snipper} \citep{Snipper2007}, which is an all-or-nothing classifier for biogeographical ancestry. It consists of 56 SNPs, which nearly fits with the set found in \cite{Kidd2014}. Interestingly, this set is implemented in the commercially available {\sc Verogen ForenSeq DNA Signature Kit}. On \url{http://mathgene.usc.es/snipper/}, a reference database for these SNPs can be used or downloaded. Here, the reference (training) database consists of 654 and 520 individuals, the majority coming from the 1000 genomes dataset  (abbreviated 1k, \citealp{Auton2015}), which in total covers more than 3000 individuals, and the HGDP-CEPH line \citep{HGDP2002} (abbreviated HGDP). To be more specific, we use the {\sc Forenseq 56 grid} available from the above webpage (accessed February 1, 2022). Here, the sheet {\em Snipper Reference Grid} comes with 654 samples from seven continental groups: 108 Africans (AFR) from Yoruba (YRI, 1k); 79 (native) Americans (AMR) comprised of 20 Brazils (Surui, HGDP), 7 Colombians from Colombia (HGDP), 18 Peruvians (PEL, 1k) and 34 Maya from Mexico (HGDP); 103 East-Asian (EAS) Han Chinese from Beijing (1k); 99 Europeans (EUR), i.e.\ Utah residents with Western and Northern European ancestry (CEU, 1k); 134 from the Middle East (MEA; {\color{black}42 Druze, 46 Palestinian, 46 Bedouin,} HGDP); 28 Oceanians (OCE) from Papua New Guinea and 103 South-Asian (SAS) Gujarati Indians in Houston (1k). For the test data, we use the {\em SGDP Test Samples} from the same file, which were collected within the Simons Genome Diversity Project \citep{SGDP}. This data comes with population labels as well, but this information does not enter the analysis. 

Results on some test samples are presented in Figure~\ref{fig5}. We use the same representation of the results as in Figure~\ref{fig1}, based on Remark~\ref{rem:MNinfty}, i.e.\ we take into account both, finite $M$ (size of the marker set) and $N$ (size of the reference database). In Figure~\ref{fig5}(A)--(D), we see samples estimated to have ancestry in two different populations. It is known that European, Middle-East and South-East Asian ancestry are harder to distinguish than other other pairs \citep{pmid34545122}. Consequently, these distinctions have larger error bars in \ref{fig5}(A)--(B) than the African sample in \ref{fig5}(C) or the South-East-Asian sample from \ref{fig5}(D). In \ref{fig5}(E), a Siberian sample is studied. The nearest population in the reference database is either Admixed Americans or South-East Asians, which are also estimated for the sample. However, both come with a large uncertainty estimate, such that the ancestry of this sample is hard to pin down using the reference database at hand. Last, \ref{fig5}(F) shows a South-Asian Sample, and it is estimated to have ancestry from several other populations. However, error bars cover the whole non-South-East-Asian part, so these fractions may be misleading; compare with Figure~\ref{fig3}. In order to analyse, which pairs of populations the AIMset can distinguish with low uncertainty, we analysed the covariance matrix from \eqref{eq:6} for samples with equal fractions from two populations. For a sample from populations $k$ and $l$, we display $\mathbb{COV}[\hat q_k, \hat q_l]$ in Figure~\ref{fig6}. Again, we see that samples which are admixed from Middle East and Europe, as well as from Middle East and South-East Asia, will have the largest error bars. 

\begin{figure}
\includegraphics[width=0.95\textwidth]{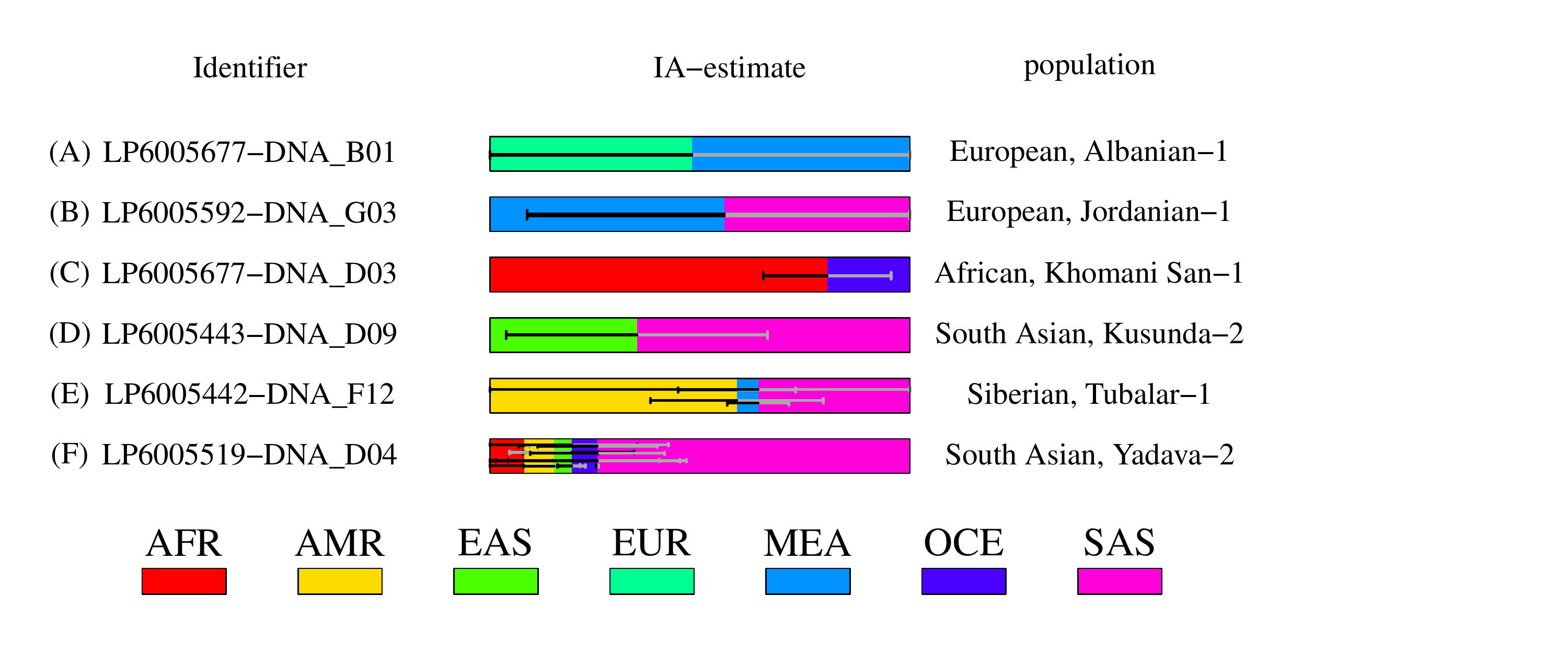}
\caption{For six samples of the Simons Genome Diversity Project \citep{SGDP}, we present the results for our analysis. In the reference database, there are seven population labels: AFR (African), AMR (Admixed American), EAS (East-Asian), EUR (European), MEA (Middle-East), OCE (Oceanian) and SAS (South-East-Asian); See the beginning of Section~\ref{S33} for more details.}
  \label{fig5}
\end{figure}

\begin{figure}
\begin{center}
\includegraphics[width=0.45\textwidth]{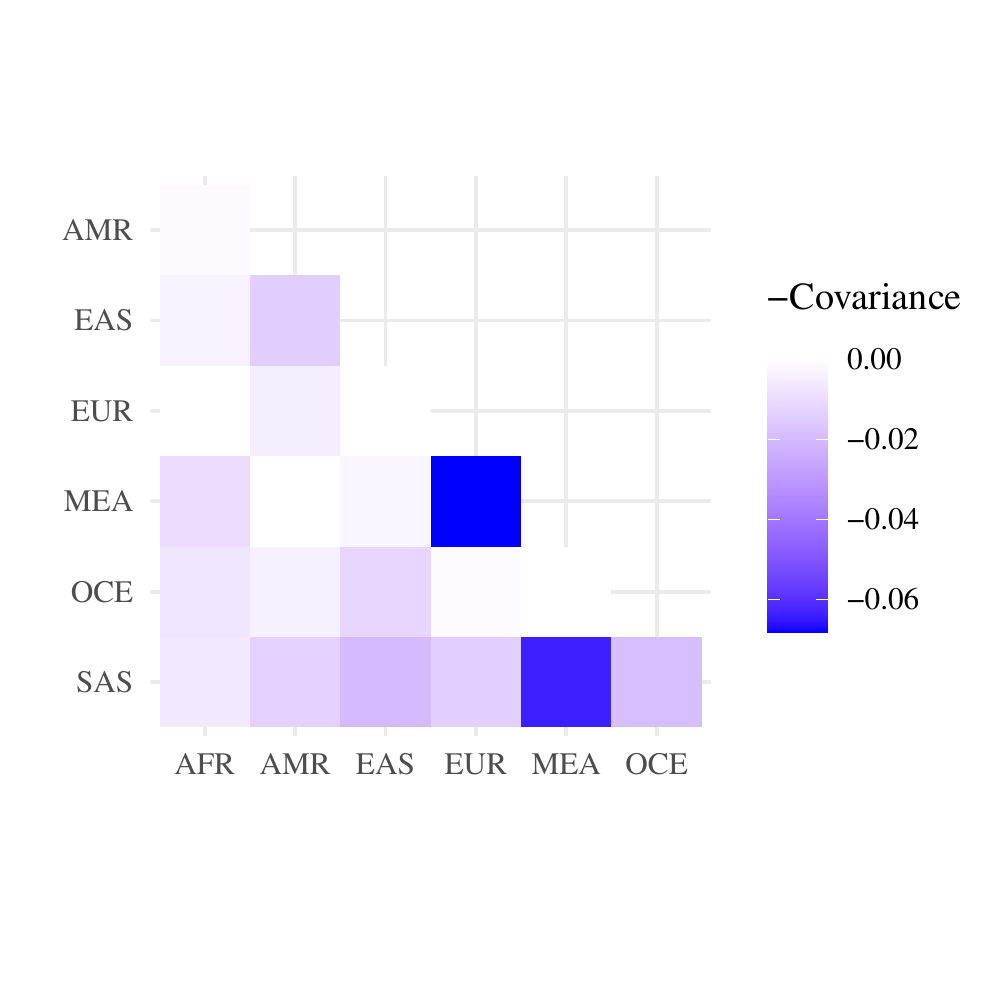}
\end{center}
\caption{When analysing a 50/50 admixed individual from populations $k, l$, the covariance of $\hat q_k$ and $\hat q_l$ is displayed.}
  \label{fig6}
\end{figure}

\subsection{Comparison to {\sc admixture}}\label{ss34}
We compared estimators for $q$, as obtained through Theorem~\ref{T2} to results obtained from {\sc admixture} (version 1.3.0). Note that the latter software has two modes which are comparable to the approach we take here. In the {\em supervised mode}, a subset of the data is allowed to have known (all-or-nothing-)ancestries. In this mode, we combine training and test data, and provide the ancestries of the training data. Here, the test sample is still involved in the estimation of allele frequencies $p$. In the {\em projection} mode, we can run {\sc admixture} with fixed allele frequencies. Here, we provide the genetic data from the test data, together with allele frequencies from the training data, which is the same case as treated here. In the comparison, we restrict ourselves to the 56 bi-allelic SNPs from \citep{Kidd2014}. Note that {\sc admixture} uses a block-relaxation algorithm and quadratic programming in order to maximize the log-likelihood $\ell$ from \eqref{eq:L|P}, while we use the iterative approach of Theorem~\ref{T2}, implemented using {\tt R}. The latter results in very compact code, but {\sc admixture} is much faster in estimating $q$. However, the number of markers in our analysis is limited, hence runtimes are no concern. For the point estimators, we find a deviation (using a total variation distance on $\mathbb S_K$, averaged over all test samples) between $\hat q$ as obtained from Theorem~\ref{T2} and to the supervised mode in {\sc admixture} of $3.1\%$, and to the projected mode of $0.026\%$. The smaller difference for the projected mode makes sense since this model follows the statistical model from Definition~\ref{def:admixtureModel}. When comparing the estimators of the variance of $\hat q$, {\sc admixture} gives bootstrap estimates for the bias and standard error of all $\hat q_k$. However, it is not described if both, individuals and markers are resampled, or only one of them. In addition, covariances between $\hat q_k$ and $\hat q_\ell$ are not reported for $k\neq \ell$. 



\section{Discussion}
{\color{black}Estimating individual admixture from individual genetic data is based on the admixture model, which can only be a coarse approximation of reality; see e.g.\ \cite{pmid30108219, pmid32323416}. Still, under ideal conditions, i.e.\ if the data follows the admixture model, the estimation of individual admixture comes with several statistical questions we are dealing with. Here, we present results on the effects of choice of markers, choice of reference populations, and finiteness of reference database on uncertainties in the estimation process.}


Barplots, used as graphical illustrations for point estimates of IA, frequently have to be taken with caution. For example,  \cite{pmid30108219} show that unsampled populations can lead to spurious results. Along these lines, we show (Figure~\ref{fig3}) that a large number of reference populations may lead to higher levels of admixture estimates. One way out is regularization and to penalize admixture per se; see \citep{pmid21682921}. Another method is to use a Bayesian approach and use an apriori distribution which is highly concetrated on non-admixed samples. The resulting posteriori distribution will then put more weight on less admixed samples as well. The method we choose here -- see e.g.\ Figure~\ref{fig1} -- is to put error bars based on theoretical results on top of a barplot, which gives an estimate of the variance of the estimator. One advantage is an easy to interpret overview of uncertaintly in the estimation of IA. As an example, recall that allele frequencies of Europeans and individuals from the Middle-East are similar \citep{pmid34545122}, which may lead to wrong assignment. Here, when using error bars -- see Figure~\ref{fig5} -- we see at least that point estimates for IA come with a high degree of uncertainty. 

Technically, our analysis of the approximative normality of the Maximum-Likelihood estimator (Theorem~\ref{T1}) is dealing with a simplified population model, where geographic origins of individuals from the reference database are known. The reason for this simplification is that we are treating the supervised case where a reference database with known ancestry is available, which is in contrast to the unsupervised case implemented in {\sc STRUCTURE} or {\sc ADMIXTURE}, where allele frequencies in all groups $p$ are estimated simultaneously with $q$ for all datasets. We are restricting ourselves to the simpler supervised case for two reasons: First, analytical results only seem to be available if $p$ is known, and second, in forensic applications, the test and training/reference data are two independently collected datasets.

In the analysis of the admixture model, we have to impose some minimal assumptions. The most severe is that $q_k>0$ for all $k$, i.e.\ all ancestral populations have contributed to the sample. (This is $q^\ast \in \mathbb S^\circ$ in Lemma~\ref{l:1} and $\hat q\in \mathbb S^\circ$ in Theorem~\ref{T1}.) Although this is not the case in many examples, our results from Theorem~\ref{T1} can still be used on the subset of populations which contribute to the test sample. On a qualitative level, we show (see Figure~\ref{fig7}) that it is more efficient {\color{black} to increase} the number of markers -- provided they are able to separate populations -- than the size of the reference database in order to reduce uncertainty in IA estimation. However, what our results also show is that it is not always safe to take small error bars as a proof of a good model fit -- see Figure~\ref{fig4}. In order to obtain an idea how well the model fits the data, \cite{pmid32323416} study correlations of the deviation from the genetic data to their expectation (depending on IA) along loci. Another way is to implement an option saying that the reference database is inconclusive about the test sample; see \cite{Tvedebrink2018}. Another would be to study the distribution of $\hat q$ under misspecification of the model, where the test data is generated under a model involving more populations than used for obtaining the IA. Studying all these sources of randomness for IA inference will eventually lead to more robust results.

\subsubsection*{Acknowledgements}
We thank two anonymous referees and the editor for a very close look at our manuscript, which led to many improvements. {\color{black} PP and AR thank the Freiburg Center for Data analysis and Modeling for partial funding.}

\end{document}